\documentclass[10pt, a4paper]{amsart}
\usepackage{latexsym, amsfonts, euscript, amsmath, amssymb, amsfonts}
\newtheorem{theorem}{Theorem}
\newtheorem{proposition}[theorem]{Proposition}
\newtheorem{lemma}[theorem]{Lemma}
\newtheorem{corollary}[theorem]{Corollary}

\newtheorem*{theorem*}{Theorem}
\newtheorem*{conjecture*}{Conjecture}
\newtheorem{definition}[theorem]{Definition}
\theoremstyle{remark}
\newtheorem{remark}[theorem]{Remark}

\newtheorem{example}[theorem]{Example}

\newcommand{\GL}{{\mathrm{GL} }}
\newcommand{\w}{{\mathrm{w_H}}}

\newcommand{\Aut}{{\mathrm{Aut}}}

\newcommand{\SL}{{\mathrm{SL}}}
\newcommand{\stab}{{\mathrm{Stab}}}
\newcommand{\Aff}{\mathbb{A}}
\newcommand{\PP}{\mathbb{P}}
\newcommand{\Z}{\mathbb{Z}}

\newcommand{\Adelta}{\mathbb{A}^{\delta}\left({\mathbb F}_q\right)}
\newcommand{\F}{\mathbb{F}_q}
\newcommand{\Fq}{\mathbb{F}_q}
\newcommand{\Fqm}{\mathbb{F}_q^m}
\newcommand{\C}{C^\mathbb{A}(\ell,m)}
\newcommand{\Ev}{\mathrm{Ev}}
\newcommand{\rank}{\mathrm{rank}}
\newcommand{\supp}{\mathrm{supp}}

\begin{document}

\title{Affine Grassmann Codes}

\author{Peter Beelen}
\address{Department of Mathematics, Technical University of Denmark, \newline \indent
DK 2800, Lyngby, Denmark.}
\email{p.beelen@mat.dtu.dk}

\author{Sudhir R. Ghorpade}
\address{Department of Mathematics, 
Indian Institute of Technology Bombay,\newline \indent
Powai, Mumbai 400076, India.}
\email{srg@math.iitb.ac.in}

\author{Tom H{\o}holdt}
\address{Department of Mathematics, Technical University of Denmark, \newline \indent
DK 2800, Lyngby, Denmark.}
\email{T.Hoeholdt@mat.dtu.dk}

\date{October 8, 2009; Revised: June 10, 2010. To appear in the \emph{IEEE Transactions of Information Theory}, Vol. 56, No. 7 (July 2010).} 

\begin{abstract}
We consider a new class of linear codes, called affine Grassmann codes. These can be viewed as a variant of generalized Reed-Muller codes and are closely related to Grassmann codes. We determine the length, dimension, and the minimum distance of any affine Grassmann code. Moreover, we show that 
affine Grassmann codes have a large automorphism group and determine the number of minimum weight codewords.
\end{abstract}
\maketitle

\section{Introduction}

Reed-Muller codes are among the most widely studied classes of linear error correcting codes. 
Numerous generalizations and variants of Reed-Muller codes have also been considered in the literature. (See, for example, \cite{DGM}, \cite[Ch. 13--15]{MW}, \cite[Ch. 1, \S 13; Ch. 11, \S 3.4.1; Ch. 16, \S 3; Ch. 17, \S 4]{handbook} and the relevant references therein). In this paper we introduce a class of linear codes that appears to be 
a genuinely distinct variant of Reed-Muller codes. As explained in Section \ref{sec:grassmann}, this new class of codes is intimately related to the so-called Grassmann codes, which have been of much current interest (see, for example, \cite{GL, GPP, HJR2, N} and the relevant references therein), 
and with this in view we call these the \emph{affine Grassmann codes}. Roughly speaking, affine Grassmann codes are obtained by evaluating 
linear polynomials in the minors of a generic 
$\ell \times \ell'$ matrix at all points of the corresponding affine space over a finite field. Evidently, when $\ell =1$, this gives the first 
order generalized Reed-Muller code ${\rm RM}(1,\ell')$. However, in general, the resulting code is distinct from higher order generalized Reed-Muller codes and determination of several of its properties appears to be rather nontrivial.  Our main results include the determination of the minimal distance (Theorem \ref{thh:mindist}) and a characterization as well as an explicit enumeration of the minimum weight codewords (Theorems \ref{thm:char} and \ref{thm:enum}). Further, we show that affine Grassmann codes have a large automorphism group (Theorem \ref{thm:autom}); this result could be viewed as an  extension of the work of Delsarte, Goethals and MacWilliams \cite[Thm. 2.3.1]{DGM}, Kn{\"o}rr and Willems \cite{KW} as well as Berger and Charpin \cite{BC} on the automorphisms of Reed-Muller codes. In geometric terms, some of our results could be viewed as a generalization of elementary facts about hyperplanes over finite fields to ``determinantal hyperplanes''. (See Remark \ref{rem:dethyper} for greater details.) The auxiliary results obtained in the course of proving the main theorems and the techniques employed may also be of some independent interest.   

\section{Preliminaries}
\label{sec:prelim}

Denote, as usual, by $\F$ the finite field with $q$ elements. Fix positive integers $\ell$ and $\ell'$ and a $\ell \times \ell'$ matrix $X= \left(X_{ij}\right)$ whose entries are algebraically independent indeterminates over $\F$. By $\F[X]$ we denote the polynomial ring in the  $\ell\ell'$ indeterminates $X_{ij}$ ($1\le i\le \ell, \; 1\le j \le {\ell}'$) with coefficients in $\F$.   For convenience, we introduce the following notation for the rows and columns of the matrix $X$: 
$$
{\mathbf{X}}_i=(X_{i1} \cdots X_{i\ell'}) \mbox{ for $1\le i\le \ell$} \quad \makebox{ and } \quad 
{\mathbf{X}}^j=\left(\begin{array}{c}X_{1j}\\ \vdots \\ X_{\ell j}\\\end{array}\right) \mbox{for $1\le j\le \ell'$}.
$$
Recall that by a \emph{minor} of $X$ of order $i$ we mean the determinant of an $i\times i$ submatrix of $X$. A minor of $X$ of order $i$ is 
sometimes referred to as an $i \times i$ minor of $X$. But in any case, it should be remembered that the minors of $X$ are not matrices, but are 
elements of the polynomial ring $\F[X]$. 

We are primarily interested in the linear space generated by all the minors of $X$. This is unchanged if we replace $X$ by its transpose. 
With this in view, we shall always assume that $\ell \le \ell'$. Further, we set
$$
m = \ell + \ell' \quad \mbox{and} \quad \delta = \ell \ell'.
$$
For $0\le i\le \ell$, we let $\Delta_i(\ell,m)$ be the set of all $i \times i$ minors of $X$, where, as per standard conventions, the only $0\times 0$ minor of $X$ is $1$.  We define
$$
\Delta(\ell,m) = \bigcup_{i=0}^{\ell} \Delta_i(\ell,m) .
$$

\begin{definition} 
\label{def:Flm}
The linear space ${\mathcal F}(\ell,m)$ over $\F$ is the  subspace of $\F[X]$ generated by $\Delta(\ell,m)$.
\end{definition}
%

For example, if $\ell=\ell'=2$, then $m=4$, $\delta=4$, and 
$\Delta_0(2,4)=\{1\}$, while 
$$
\Delta_1(2,4)=\{X_{11},X_{12},X_{21},X_{22}\} \text{ and } \Delta_2(2,4)=\{X_{11}X_{22}-X_{12}X_{21}\}.
$$
Thus a typical element of ${\mathcal F}(2,4)$ looks like
\begin{equation}
\label{TypicalF24}
a + b_1 X_{11} + b_2 X_{12} + b_3 X_{21} + b_4 X_{22} + c\left(X_{11}X_{22}-X_{12}X_{21}\right)
\end{equation}
where $a,b_1,b_2,b_3,b_4,c\in \Fq$. Observe that $\# \Delta(2,4) = 6$, where
for a finite set $D$, we 
denote by $\# D$ the cardinality of $D$. In general, we have the following. 

\begin{lemma}
\label{lem:CardDeltaml}
The cardinality of $\Delta(\ell,m)$ is ${\binom{m}{\ell}}$.
\end{lemma}

\begin{proof}
Since the entries of $X$ are indeterminates, the number of minors of $X$ of order $i$ is the number of $i\times i$ submatrices of $X$. 
An  $i\times i$ submatrix of $X$ is obtained by choosing 
$i$ rows among the $\ell$ rows and $i$ columns among the $\ell'$ columns. Thus 
$$
\# \Delta_i(\ell,m) = {\binom{\ell}{i}}{\binom{\ell'}{i}}  \quad \mbox{ for $0\le i\le \ell$}.
$$
Consequently,  
$$
\# \Delta(\ell,m) = \sum_{i\ge 0} {\binom{\ell}{\ell-i}}{\binom{\ell'}{i}} 
={\binom{m}{\ell}},
$$ 
where the last equality 
follows from the so-called Chu--Vandermonde summation 
(see e.g\@. \cite[Sec.~5.1, (5.27)]{GrKPAA}). 
\end{proof}

We remark that an alternative proof of the above lemma 
can be obtained by observing that the minors of $X$ (of arbitrary orders) are in a natural one-to-one correspondence with the $\ell \times \ell$ minors of the $\ell\times m$ matrix $(X|I)$ obtained by adjoining to $X$ a $\ell \times \ell$ identity matrix.

%

The following basic result can be viewed as a very special case of the standard basis theorem or the straightening law 
of Doubilet, Rota and Stein (cf. \cite{DRK1974}, \cite[Thm. 4.2]{ghorpade1994}). In the case we are interested in, a much 
simpler proof can be given and this is included below. 

\begin{lemma}
\label{lem:dimFml}
The elements of $\Delta(\ell,m)$ are linearly independent. In particular,  
$$
\dim_{\F}{\mathcal F}(\ell,m)=\binom{m}{\ell}.
$$
\end{lemma}
\begin{proof}
Suppose there is a linear dependence relation $\sum_{{\mathcal M} \in \Delta(\ell,m)}a_{\mathcal M}{\mathcal M}=0$, 
where  $a_{\mathcal M}\in \F$ for ${\mathcal M}\in \Delta(\ell,m)$. 
We will show by finite induction on $i$ ($0\le i \le \ell$) that $a_{\mathcal M} =0$ for all ${\mathcal M}\in \Delta_i(\ell,m)$. First, by specializing all the variables to zero (i.e., by substituting $X_{rs}=0$ for all $ r \in\{1,\dots , \ell\} $ and $s \in\{1. \dots , \ell'\}$ 
in the linear dependence relation), we see that the desired assertion holds when $i=0$. Next, suppose $i>0$ and $a_{\mathcal M} =0$ for all ${\mathcal M}\in \Delta_j(\ell,m)$ and all $j<i$.  
Pick a minor ${\mathcal M} \in \Delta_i(\ell,m)$. By specializing all the variables except the ones occurring in ${\mathcal M}$ to zero, we obtain 
$a_{\mathcal M} =0$. 
Repeating this procedure for each $i \times i$ minor, we obtain the induction step. This proves that the elements of $\Delta(\ell,m)$ are linearly independent. Consequently,  $\dim_{\F}{\mathcal F}(\ell,m) = \# \Delta(\ell,m) = {\binom{m}{\ell}}$. 
\end{proof}

\medskip

Thanks to Lemma \ref{lem:dimFml}, every element of ${\mathcal F}(\ell,m)$ is a unique $\F$-linear combination of the elements of $\Delta(\ell,m)$. 
With this in view, we make the following definition.  

\begin{definition}\label{def:supp}
Given $f=\sum_{{\mathcal M} \in \Delta(\ell,m)}a_{\mathcal M}{\mathcal M}\in{\mathcal F}(\ell,m)$, where $a_{\mathcal M}\in \F$ for every ${\mathcal M}\in{\Delta}(\ell,m)$, 
the support of $f$ is the set 
$$\supp(f):=\{{\mathcal M}\in \Delta(\ell,m) \, : \, a_{{\mathcal M}} \neq 0\}.$$
Note that the the support of $f$ is the empty set if and only if $f$ is the zero polynomial. 
\end{definition}
 

We shall denote the space of all $\ell\times \ell'$ matrices with entries in $\F$ by $\Adelta$, or simply by ${\mathbb A}^{\delta}$. Indeed,
for fixed positive integers $\ell$ and $\ell'$, this space can be readily identified with the $\delta$-dimensional affine space over $\Fq$,
where $\delta = \ell \ell'$, as before. 
It is clear that for any $f\in \F[X]$ (and in particular, any $f\in {\mathcal F}(\ell,m)$) 
and $P\in \Aff^{\delta}$, the element $f(P)$ of $\F$ is well-defined. Now let us fix an enumeration $P_1,P_2,\dots,P_{q^\delta}$ of $\Aff^{\delta}$. 

\begin{definition} The evaluation map of $\F[X]$ is the map 
$$
\Ev: \F[X]\to  \F^{q^{\delta}} \quad \mbox{ defined by } \quad 
\Ev(f):= \left(f(P_1),\dots,f(P_{q^\delta}) \right).
$$ 
\end{definition}

It is clear that the evaluation map $\Ev$ defined above is a surjective linear map. Also, it is well-known that the kernel 
of $\Ev$ is 
the ideal of $\F[X]$ generated by $\left\{ X_{ij}^q-X_{ij} \, : \, 1\le i \le \ell,\; 1\le j \le \ell' \right\}$, 
and that this kernel contains no
nonzero polynomial having degree $< q$ in each of the variables.  (See, for example, \cite[p. 11]{Joly}.) 
In particular, if $0\ne f \in {\mathcal F}(\ell,m)$, then $f$ cannot be in the kernel of $\Ev$ because $\deg_{X_{ij}}(f)\le 1$ for each variable $X_{ij}$. Thus the restriction of the evaluation map $\Ev$ to ${\mathcal F}(\ell,m)$ is injective. We are now ready to define the codes that are  studied in the remainder of this paper.

\begin{definition}
The affine Grassmann code $\C$ is 
the image of ${\mathcal F}(\ell,m)$ under the evaluation map $\Ev$.  
The minimum distance of $\C :=\Ev\left({\mathcal F}(\ell,m)\right)$ will be denoted by $d(\ell,m)$.
\end{definition}

Recall that a code $C$ is said to be \emph{degenerate} if there exists a coordinate position $i$ such that $c_i=0$ for all $c\in C$. 
It turns out that affine Grassmann codes are nondegenerate and their length and dimension are easily determined.

\begin{lemma}
The affine Grassmann code $\C$ is a nondegenerate linear code of length $q^{\delta}$ and dimension $\binom{m}{\ell}$.
\end{lemma}
\begin{proof}
It is obvious that $\C$ is a linear code of length $q^\delta = \# {\mathbb A}^{\delta}(\F)$. Moreover, since the constant polynomial $1$, being the only element of $\Delta_0(\ell,m)$, is in ${\mathcal F}(\ell,m)$,  and since  
$\Ev(1)=(1,\dots, 1)$, it follows that $\C$ is nondegenerate. 
Finally, since the evaluation map is injective on ${\mathcal F}(\ell,m)$, it follows from Lemma \ref{lem:dimFml} that the dimension of $\C$ is  $\binom{m}{\ell}$.
\end{proof}

\begin{example}
\label{exa:WhatTheRefereeWanted}
Suppose $\ell = \ell'=2$. Then $m=\delta =4$ and the elements of ${\mathcal F}(2,4)$ are of the form \eqref{TypicalF24}. The affine space ${\mathbb A}^{\delta}$ consists of the $2\times 2$ matrices with entries in ${\mathbb F}_2=\{0,1\}$. There are $16$ such matrices and, upon letting
 $e_{ij}$ denote the $2\times 2$ matrix with $1$ in $(i,j)^{\rm th}$ position and $0$ elsewhere, ${\mathbb A}^{\delta}({\mathbb F}_2)$ may be enumerated as 
$
\mathbf{0}, \; e_{11}, \; e_{12}, \; e_{21}, \; e_{22}, \; e_{11} + e_{12}, \; e_{11} + e_{21}, \; e_{11} + e_{22}, \; e_{12} + e_{21}, \; e_{12} + e_{22}, \; e_{21} + e_{22}, \;  
e_{11} + e_{12} + e_{21}, \; e_{11} + e_{12} + e_{22}, \; e_{11} + e_{21} + e_{22}, \; e_{12} + e_{21} + e_{22}, \; e_{11} + e_{12} + e_{21} + e_{22},
$
where $\mathbf{0}$ denotes the $2\times 2$ zero matrix. Accordingly, the codewords of $C^{\mathbb{A}}(2,4)$ consist of the elements of ${\mathbb F}_2^{16}$ of the form $a\mathbf{1} + \mathbf{v}$, where $\mathbf{1}$ denotes the $16$-tuple all of whose coordinates are $1$, whereas  
$\mathbf{v}$ denotes the $16$-tuple given by $(0, \; b_1, \; b_2, \; b_3, \; b_4, \; b_1+b_2,  \; b_1+b_3,  \; b_1+b_4+c,  \; b_2+b_3+c,  \; b_2+b_4,  \; b_3+b_4,  \; b_1+b_2+b_3+c, \; b_1+b_2+b_4+c, 
\; b_1+b_3+b_4+c, \; b_2+b_3+b_4+c,    \; b_1+b_2+b_3+b_4+c)$.
Here $a,b_1,b_2,b_3,b_4, c$ vary over ${\mathbb F}_2$. As such, there are $2^6 = 64$ codewords, and it is clear that the code is nondegenerate and its dimension is $6$; indeed, a ${\mathbb F}_2$-basis of $C^{\mathbb{A}}(2,4)$ is obtained by setting exactly one of $a,b_1,b_2,b_3,b_4, c$ to be $1$ and the others to be $0$.  
Further, by listing the $64$ codewords, it is easily seen that every nonzero codeword is of (Hamming) weight $\ge 6$, and the codeword corresponding $a=b_1=b_2=b_3=b_4=0$ and $c=1$ is of weight $6$. Thus, at least in the binary case, $C^{\mathbb{A}}(2,4)$ is a $[16,6,6]$-code. 
\end{example}

We end this section by giving two lemmas on determinants that will be useful in the sequel. 

\begin{lemma}\label{lem:xplusa}
Let $Y= \left(Y_{ij}\right)$ be a $\ell \times \ell$ matrix whose entries are independent indeterminates over $\F$ 
and let $B=\left(b_{ij}\right)$ be a $\ell \times \ell$ matrix with entries 
in $\F$. Then there is $h\in {\mathcal F}(\ell,2\ell)$ such that 
$$
\det(Y+B)=\det(Y)+\sum_{1\le i,j\le \ell}(-1)^{i+j}b_{ij}\det(Y^{ij})+h \quad \mbox{with} \quad \supp(h) \subseteq  \bigcup_{i=0}^{\ell-2}\Delta_i(\ell,2\ell), 
$$ 
where $Y^{ij}$ denotes the $(\ell -1)\times (\ell -1)$ 
matrix obtained from $Y$ by deleting the $i$-th row and the $j$-th column.
\end{lemma}
\begin{proof}
For a subset $S$ of $\{1,\dots,\ell\}$, denote by $(Y,B)_S$ the matrix obtained from $Y$ by replacing for all $j \in S$, the $j$-th column of $Y$ by the $j$-th column of $B$. By the multilinearity of the determinant, we readily see that 
$$
\det(Y+B)=\sum_{S}\det((Y,B)_S),
$$
where the sum is over all subsets $S$ of $\{1,\dots,\ell\}$. 
Observe that if $S$ is the empty set, then $\det((Y,B)_S)=\det(Y)$. Moreover, if $S$ is singleton, say $S=\{j\}$, where $1\le j \le \ell$, then 
by developing the determinant along the $j$-th column we find that 
$$
\det((Y,B)_S)=\sum_{i=1}^{\ell}(-1)^{i+j}b_{ij}\det(Y^{ij}).
$$
Finally, if $S\subseteq \{1,\dots,\ell\}$ with $\# S = s \ge 2$, then using Laplace expansion along the columns indexed by the elements of $S$, we 
see that $\det((Y,B)_S)$ is a $\F$-linear combination of minors in $\Delta_{\ell-s}(\ell,2\ell)$. This yields the desired result. 
\end{proof}

\medskip

We will also need the following well-known result whose proof can be found, for example, in \cite[Ch. I, \S 2]{Gant}. 

\begin{lemma}[Cauchy-Binet]\label{lem:CB}
Let $r$ and $s$ be positive integers such that $r\le s$, and let $A$ be a $r \times s$ matrix and $B$ a $s \times r$ matrix with entries in a commutative ring. 
For a subset $I$ of $\{1,\dots,s\}$ with $\# I = r$,  denote by $A^I$ the $r\times r$ submatrix of $A$ formed by 
the $j$-th columns of $A$ for $j \in I$, and denote by $B_I$ the $r\times r$ submatrix of $B$ formed by 
the $i$-th rows of $B$ for $i \in I$. Then 
$$
\det(AB)=\sum_{I}\det(A^I)\det(B_I),
$$ 
where the sum is over all subsets $I$ of $\{1,\dots, s\}$ of cardinality $r$.
\end{lemma}

\begin{remark}\label{rem:dethyper}
As a warm-up for the results of the subsequent section, let us consider the 
case of $\ell=1$ even though it is rather trivial. 
Here ${\mathcal F}(1,m)$ corresponds to the space of linear polynomials in $\ell'$ variables of the 
form $h = a_0 + a_1X_{11} + \cdots + a_{\ell'}X_{1\ell'}$. 
For any such $h$, the Hamming weight of the corresponding codeword $\Ev(h)$ amounts to finding the the number of $\F$-rational points on a hyperplane in $\Aff^{\ell'}$. Indeed, assuming that $\Ev(h)$ is nonzero, or equivalently that not all $a_0, a_1, \dots , a_{\ell'}$ are zero, it is readily seen that
$$
\w\left(\Ev(h)\right) \; = \; \# \Aff^{\ell'}(\F)  - \# H  \; = \; \begin{cases} q^{\ell'} & \text{if } a_1= \dots = a_{\ell'}=0, \\ q^{\ell'} - q^{\ell'-1} & \text{otherwise,} \end{cases}
$$
where $H$ denotes the affine hyperplane $\{P\in \Aff^{\ell'} (\F) : h(P)=0\}$. 
It follows that the minimum distance of $C^{\Aff}(1, m)$ is $q^{\ell'-1} (q-1)$,
and also that the number of minimum weight codewords is $(q^{\ell'+1} - q)$. 
In a similar manner, the general case corresponds to finding the maximum number of points on a ``determinantal hyperplane'', i.e., the zero-set of an arbitrary nonzero element of ${\mathcal F}(\ell,m)$, and finding the minimum weight codewords corresponds to finding those determinantal hyperplanes where the maximum is attained. 
\end{remark}

\section{Minimum distance}
\label{sec:mindist}
In this section we will compute the minimum distance $d(\ell,m)$ of the affine Grassmann code $\C$. We start by determining the Hamming weight of a maximal minor, obtaining thereby an upper bound for $d(\ell,m)$. As usual we denote by $\w(c)$ the Hamming weight of a codeword $c$.

\begin{lemma}\label{lem:upper}
Let $\mathcal M \in \Delta_{\ell}(\ell,m)$. Then 
$$
\w(\Ev({\mathcal M}))=q^{\delta-\ell^2}\prod_{i=0}^{\ell-1}(q^\ell-q^i).
$$
In particular, 
$$
d(\ell,m)\le q^{\delta-\ell^2}\prod_{i=0}^{\ell-1}(q^\ell-q^i).
$$
\end{lemma}
\begin{proof}
Without loss of generality we shall assume that ${\mathcal M}$ is the leading maximal minor, i.e., ${\mathcal M}=\det((X_{ij})_{1\le i,j\le \ell})$. 
Let $P \in \Adelta$ and let $(p_{ij})_{1\le i \le \ell, \; 1\le j \le \ell'}$ be the $\ell \times \ell'$ matrix with entries in $\F$ corresponding to $P$. It is  clear that ${\mathcal M}(P)\neq 0$ if and only if the $\ell \times \ell$ submatrix $(p_{ij})_{1\le i,j \le \ell}$ is nonsingular. This happens for exactly $\prod_{i=0}^{\ell-1}(q^\ell-q^i)$ values of $p_{ij}$ with $1\le i,j \le \ell$. The remaining $\ell\ell' - \ell^2$ values $p_{ij}$ with $j>\ell$ do not play any role in the evaluation of ${\mathcal M}$ at $P$. Hence 
$\w(\Ev({\mathcal M}))=q^{(\delta-\ell^2)}\prod_{i=0}^{\ell-1}(q^\ell-q^i)$. This 
implies the desired inequality for $d(\ell,m)$. 
\end{proof}

\medskip

We will show that the upper bound for $d(\ell,m)$ in the above lemma gives, in fact, the true minimum distance. 
To this end, the specialization maps defined below will be useful. 

\begin{definition}\label{def:spec}
Let $i, j$ be integers satisfying $1 \le i \le \ell$ and $1 \le j \le \ell'$,  and let $\mathbf{a} = \left(a_{1}, \dots , a_{\ell'}\right) \in \F^{\ell'}$ and $\mathbf{b} = \left(b_{1}, \dots , b_{\ell}\right) \in \F^{\ell}$. The row-wise specialization map relative to $\mathbf{a}$ and $i$ is the map
$$
s_{\mathbf{a}}^{(i)}:  {\mathcal F}(\ell,m)   \to  {\mathcal F}(\ell-1,m-1) \quad \mbox{defined by} \quad 
s_{\mathbf{a}}^{(i)}(f):= f\left|_{{\mathbf{X}}_i={\mathbf{a}}},\right.
$$ 
i.e., $s_{\mathbf{a}}^{(i)}(f)$ is the element of ${\mathcal F}(\ell-1,m-1)$ obtained by substituting $\left(X_{i1}, \dots , X_{i\ell'}\right) = \left(a_{1}, \dots , a_{\ell'}\right)$ in 
$f\left(X_{11}, \dots , X_{\ell\ell'}\right)$. 
Further, if $\ell'>\ell$, then 
the column-wise specialization map relative to $\mathbf{b}$ and $j$ is the map
$$
t_{\mathbf{b}}^{(j)}:  {\mathcal F}(\ell,m)   \to  {\mathcal F}(\ell,m-1) \quad \mbox{defined by} \quad 
t_{\mathbf{b}}^{(j)}(f):= f\left|_{{\mathbf{X}}^j={\mathbf{b}}},\right.
$$ 
i.e., $t_{\mathbf{b}}^{(j)}(f)$ is the element of ${\mathcal F}(\ell,m-1)$ obtained by substituting $\left(X_{1j}, \dots , X_{\ell j}\right) = \left(b_{1}, \dots , b_{\ell}\right)$ in $f$. 
It may be noted that $s_{\mathbf{a}}^{(i)}$ and  $t_{\mathbf{b}}^{(j)}$ are $\F$-linear maps. 
\end{definition}

%

\begin{lemma}\label{lem:spec}
Let $f \in {\mathcal F}(\ell,m)$ and let $i, j$ be integers satisfying $1 \le i \le \ell$ and $1 \le j \le \ell'$. 
Then 
\begin{equation}
\label{eq:RowSpecWt}
\w(\Ev(f))=\sum_{{\mathbf{a}} \in \F^{\ell'}}\w\big(\Ev(s_{\mathbf{a}}^{(i)}(f))\big).
\end{equation}
Moreover, if $\ell'>\ell$, then 
\begin{equation}
\label{eq:ColSpecWt}
\w(\Ev(f))=\sum_{{\mathbf{b}} \in \F^{\ell}}\w\big(\Ev(t_{\mathbf{b}}^{(j)}(f))\big).
\end{equation}
\end{lemma}
\begin{proof}
Given any $\mathbf{a}\in \F^{\ell'}$, the specialization $s_{\mathbf{a}}^{(i)}(f)$ is in ${\mathcal F}(\ell-1, m-1)$ and hence the codeword
$\Ev(s_{\mathbf{a}}^{(i)}(f))$ has $q^{\delta-\ell'}$ coordinates;  
each of these coordinates can be computed by evaluating $f$ at those points $P=(p_{ij})$ of $\Adelta$ satisfying  $(p_{i1}, \dots, p_{i\ell'})= {\mathbf{a}}$. As $\mathbf{a}$ varies over 
$\F^{\ell'}$, all the $q^{\delta}$ coordinates of $\Ev(f)$ will be accounted for. Thus the codeword $\Ev(f)$ can be partitioned into shorter codewords $\Ev(s_{\mathbf{a}}^{(i)}(f))$, ${\mathbf{a}}\in \F^{\ell'}$. This implies 
\eqref{eq:RowSpecWt}. The proof of \eqref{eq:ColSpecWt} is similar. 
\end{proof}

\medskip

We shall now consider the special case $\ell= \ell'$, i.e., $m=2\ell$. 
In this case, $X$ has a unique maximal minor and whenever it occurs in a polynomial in ${\mathcal F}(\ell,2\ell)$, 
all the submaximal minors occurring in that polynomial can be killed  by a translation. 
 
\begin{lemma}\label{lem:absorb}
Let $f \in {\mathcal F}(\ell,2\ell)$ be such that $\det(X) \in \supp(f)$ and the coefficient of $\det(X)$ in $f$ equals $1$. Then there exists a unique $\ell \times \ell$ matrix $A$ with entries in 
$\F$ such that 
$$
f=\det(X+A)+h, \quad \mbox{where }  h\in {\mathcal F}(\ell,2\ell)  \mbox{ with } \; \supp(h) \, \subseteq \bigcup_{i=0}^{\ell-2}\Delta_i(\ell,2\ell).
$$
\end{lemma}
\begin{proof}
If $\ell =1$, then the desired result holds trivially with $h=0$. 
Assume that $\ell \ge 2$. 
For $1\le i,j \le \ell$, let $X^{ij}$ denote the $(\ell-1)\times(\ell-1)$ submatrix of $X$ obtained by deleting the $i$-th row and the $j$-th column, and let $b_{ij}$ denote the coefficient of $\det(X^{ij})$ in $f$.
Then 
there is  $h_1 \in {\mathcal F}(\ell,2\ell)$ such that 
$$
f=\det(X)+\sum_{1\le i,j \le \ell}b_{ij}\det(X^{ij})+h_1 \quad \mbox{and} \quad \supp(h_1) \subseteq \bigcup_{i=0}^{\ell-2}\Delta_i(\ell,2\ell). 
$$ 
Now define $a_{ij}=(-1)^{i+j}b_{ij}$ for $1\le i,j \le \ell$ and let $A$ denote the $\ell \times \ell$ matrix $(a_{ij})_{1\le i,j \le \ell}$. 
By Lemma \ref{lem:xplusa}, there is $h_2\in {\mathcal F}(\ell,2\ell)$ such that 
$$
\det(X+A)=\det(X)+\sum_{1\le i,j \le \ell}b_{ij}X^{ij}+h_2 \quad \mbox{ and } \quad \supp(h_2) \subseteq \bigcup_{i=0}^{\ell-2}\Delta_i(\ell,2\ell).
$$ 
Thus $f=\det(X+A)+h$, where $h:=h_1-h_2$, and we have the desired result. 
\end{proof}

We are now ready to prove the main result of this section. 

\begin{theorem}\label{thh:mindist}
The minimum distance $d(\ell,m)$ of the code $\C$ is given by 
\begin{equation}
\label{eq:dlm}
d(\ell,m)=q^{\delta-\ell^2}\prod_{i=0}^{\ell-1}(q^\ell-q^i).
\end{equation}
\end{theorem}
\begin{proof}
We prove the theorem by induction on $m$. Note that $m\ge 2$ since $1\le \ell \le \ell'$. If $m=2$, then $\ell=\ell'=1$ and $\delta=1$, and so \eqref{eq:dlm} follows from the observations in Remark \ref{rem:dethyper} in this case. 
Now suppose $m>2$ and the theorem is true for all codes $C^{\mathbb A}(\ell,m-1)$, with $1 \le \ell \le \lfloor (m-1)/2 \rfloor$. 
We will prove \eqref{eq:dlm} by considering separately the cases $\ell < \ell'$ and $\ell = \ell'$. 

\smallskip

{\bf Case 1:} $\ell<\ell'$. Let $f \in {\mathcal F}(\ell,m)$ and suppose $f \neq 0$. Then $\supp (f)$ is nonempty. Choose a minor ${\mathcal M} \in \supp(f)$ of the maximum possible order, say $r$, in the sense that ${\mathcal M} \in \Delta_r(\ell,m)$ and $\Delta_s(\ell,m) \cap \supp(f) =\emptyset$ for all $s>r$. Since $r\le \ell < \ell'$, there exists a column ${\mathbf{X}}^j$ of $X$ such that the variables $X_{1j},\dots,X_{\ell j}$ do not occur in $\mathcal M$. 
In particular, $t_{\mathbf{b}}^{(j)}({\mathcal M})={\mathcal M}$ for any ${\mathbf{b}} \in \F^{\ell}$. Since $\mathcal M$ is of maximum order 
in $\supp (f)$, this implies that $t_{\mathbf{b}}^{(j)}(f)$ is not the zero polynomial and therefore the codeword $\Ev(t_{\mathbf{b}}^{(j)}(f))$ is nonzero for any ${\mathbf{b}} \in \F^{\ell}$. Consequently,  by Lemma \ref{lem:spec} and the induction hypothesis, we see that 
\begin{eqnarray*}
\w(\Ev(f))  & = & \sum_{{\mathbf{b}} \in \F^{\ell}}\w(\Ev(t_{\mathbf{b}}^{(j)}(f))) \\
            & \ge & q^{\ell}d(\ell,m-1) \\
            & = & q^{\ell}q^{(\ell'-1)\ell-\ell^2}\prod_{i=0}^{\ell-1}(q^\ell-q^i) \\
             & = & q^{\delta -\ell^2}\prod_{i=0}^{\ell-1}(q^\ell-q^i).
\end{eqnarray*}
Since the above holds for any nonzero $f \in {\mathcal F}(\ell,m)$, we obtain 
$$
d(\ell,m) \ge q^{\delta-\ell^2}\prod_{i=0}^{\ell-1}(q^\ell-q^i).
$$
This inequality together with Lemma \ref{lem:upper} establishes the induction step.

\smallskip

{\bf Case 2:} $\ell=\ell'$. In this case $m=2\ell$ and 
$X$ has only one $\ell\times\ell$ minor, namely ${\mathcal L}:= \det(X)$. Let $f \in {\mathcal F}(\ell,2\ell)$ be a nonzero polynomial. 
We will distinguish two subcases 
depending on whether or not the $\ell\times\ell$ minor ${\mathcal L}$ occurs in $f$. 

\smallskip

{\bf Subcase 1:}  ${\mathcal L}\not\in\supp(f)$. 
In this event, by a similar reasoning as in Case 1, there exists a row, say the $i$-th row, such that 
$s_{\mathbf{a}}^{(i)}(f) \neq 0$ for all ${\mathbf{a}} \in \F^{\ell}$. 
Consequently,  by Lemma \ref{lem:spec} and the induction hypothesis, we see that
\begin{eqnarray*}
\w(\Ev(f))  & = & \sum_{{\mathbf{a}} \in \F^{\ell}}\w(\Ev(s_{\mathbf{a}}^{(i)}(f))) \\
            & \ge & q^{\ell}d(2\ell-1,\ell-1) \\
            & = &q^{\ell}q^{(\ell-1)\ell-(\ell-1)^2}\prod_{i=0}^{\ell-2}(q^{\ell-1}-q^i)\\
            &  = & 
            q^{\ell}\prod_{i=0}^{\ell-2}(q^{\ell}-q^{i+1}) \\
            & >  & \prod_{i=0}^{\ell-1}(q^{\ell}-q^{i}).
\end{eqnarray*}
Thus from Lemma \ref{lem:upper}, we conclude that $\Ev(f)$ cannot be a minimum weight codeword of $\C$ if $\det(X) \not\in \supp(f)$.


\smallskip

{\bf Subcase 2:}  ${\mathcal L} \in\supp(f)$. 
In this event, by Lemma \ref{lem:absorb} there exists a $\ell \times \ell$ matrix $A$ with entries in $\F$ such that $f=\det(X+A)+h$, where $h$ is a $\F$-linear combination of $i \times i$ minors of $X$ with $0\le i \le \ell-2$. If $h=0$, then $f=\det(X+A)$ and since $\Ev(f)$ is obtained by evaluating $f$ at all points of $\Adelta$,  we see that $\w(\Ev(\det(X+A)))=\w(\Ev(\det(X)))$; hence, by Lemma \ref{lem:upper}, we then find that 
$
\w(\Ev(\det(X+A)))=\prod_{i=0}^{\ell-1}(q^{\ell}-q^{i}) = d(\ell, 2\ell).
$ 
Now suppose $h \neq 0$. Then $\ell \ge 2$ and as in Case 1, we can choose a  minor ${\mathcal M} \in \supp(h)$ of maximum order, say $r$ with $r\le \ell -2$,  and find an integer $i$ with $1\le i \le \ell$ such that $s_{\mathbf{a}}^{(i)}({\mathcal M})={\mathcal M}$ for all ${\mathbf{a}} \in \F^{\ell}$. 
Since $\mathcal M$ is of maximum order in $\supp (h)$, we see that $s_{\mathbf{a}}^{(i)}(h)\neq 0$ for all ${\mathbf{a}} \in \F^{\ell}$. Also, since 
$r\le \ell -2$, the nonzero polynomial $s_{\mathbf{a}}^{(i)}(h)$ is of (total) degree at most $\ell -2$. On the other hand,  by developing the resulting determinant along the $i$-th row, we see that $s_{\mathbf{a}}^{(i)}(\det(X+A))$ 
is either the zero polynomial or a nonzero 
polynomial in $\F[X]$ of degree $\ell -1$. 
It follows that $s_{\mathbf{a}}^{(i)}(f)\neq 0$ for all ${\mathbf{a}} \in \F^{\ell}$. Now, proceeding as in Subcase 1, we see that 
$\w(\Ev(f)) > \prod_{i=0}^{\ell-1}(q^{\ell}-q^{i})$, and so 
from Lemma \ref{lem:upper} we conclude that $\Ev(f)$ cannot be a minimum weight codeword. 
%

Thus we have shown that $d(\ell,2\ell)=\prod_{i=0}^{\ell-1}(q^{\ell}-q^{i})$ and therefore established the induction step in Case 2.   This completes the proof. 
\end{proof}

\medskip

Using the $q$-factorial function $[d]_q!:=\prod_{i=1}^{d}(q^i-1)$,  
the formula \eqref{eq:dlm} for the minimum distance of $\C$ can be more compactly written as follows. 
\begin{equation}
\label{eq:dlm2}
d(\ell,m)=q^{\delta- {\binom{\ell+1}{2}}}[\ell]_q!.
\end{equation}
Note that if $\ell=1$, then the formula \eqref{eq:dlm} as well as \eqref{eq:dlm2} for $d(\ell,m)$ is in agreement with the observations in Remark \ref{rem:dethyper}. 

\begin{remark}\label{rem:mis2l}
By analyzing the proof of Theorem \ref{thh:mindist} in greater detail, one can show that if $\ell=\ell'$, then the minimum weight codewords 
of $\C$ arise precisely from nonzero constant multiples of translates of the unique maximal minor, i.e., from polynomials of the form 
$\lambda \det(X+A)$, with $0\ne \lambda \in \F$ and $A$ an $\ell \times \ell$ matrix with entries in $\F$. Consequently, the number of minimum weight codewords in $C^{\mathbb A}(\ell,2\ell)$ is
equal to  $(q-1)q^{\ell^2}$. A more general version of these results will be proved 
in Sections \ref{sec:char} and \ref{sec:enum}.  
\end{remark}

\section{Automorphisms}
\label{sec:automorph}
Recall that the \emph{(permutation) automorphism group} $ \Aut(C)$ of a code $C\subseteq \F^n$ is the set of all permutations $\sigma$ of $\{1,\dots,n\}$ such that $(c_{\sigma(1)},\dots,c_{\sigma(n)}) \in C$ for all $c=(c_1,\dots,c_n) \in C$. Evidently, $\Aut(C)$ is a subgroup of the symmetric group on $\{1,\dots,n\}$. In this section, we shall show that the automorphism groups of affine Grassmann codes are large; more precisely, we shall show that $\Aut\left(\C\right)$ contains a subgroup of order 
\begin{equation}
\label{eq:orderAut}
q^{\delta}\prod_{i=0}^{\ell-1}(q^\ell-q^i) = n \prod_{i=0}^{\ell-1}(q^\ell-q^i) = q^{\ell^2}d(\ell,m) ,
\end{equation}
where $n$ and $d(\ell,m)$ denote the length and the minimal distance of $\C$.

Denote, as usual, by $\GL_r(\F)$ the set of all invertible $r \times r$ matrices with entries in $\F$ and by $M_{r\times s}(\F)$ the set of all $r \times s$ matrices with entries in $\F$. 
Let $A \in \GL_{\ell'}(\F)$ and ${\bf u} \in M_{\ell \times \ell'}(\F)$. Define
$$
\phi_{{\bf u},A}: \Adelta \to \Adelta
$$
to be the linear transformation given by 
$$
\phi_{{\bf u},A}(P)=PA^{-1}+{\bf u} \quad \mbox{ for }  P=(p_{ij})_{1\le i \le \ell, \; 1\le j \le \ell'}\in \Adelta,
$$ 
It is clear that the transformation 
$\phi_{{\bf u},A}$ gives a bijection of $\Aff^\delta = \Adelta$ onto itself, and hence 
$\big(f(\phi_{{\bf u},A}(P))\big)_{P\in \Aff^\delta}$ will be a permutation of $(f(P))_{P\in \Aff^\delta}$ for any $f\in \F[X]$; we shall denote
this permutation $\sigma_{{\bf u},A}$. 

\begin{lemma}\label{lem:sigmauA}
Let $A \in \GL_{\ell'}(\F)$ and ${\bf u} \in M_{\ell \times \ell'}(\F)$. Then $\sigma_{{\bf u},A} \in \Aut\left(\C\right)$.
\end{lemma}
\begin{proof}
Let $r$ be any integer with $0\le r \le \ell$. In view of Lemma \ref{lem:xplusa}, 
a $r\times r$ minor of $XA^{-1}+{\bf u}$ is a $\F$-linear combination of $i\times i$ minors of $X$, where $0\le i \le r$. Consequently, 
if $f= f(X) \in {\mathcal F}(\ell,m)$, then $f(XA^{-1}+{\bf u})\in {\mathcal F}(\ell,m)$. Moreover, 
$$
\sigma_{{\bf u},A}\left(\Ev(f) \right)= \big(f(\phi_{{\bf u},A}(P))\big)_{P\in \Adelta}=\Ev\left(f(XA^{-1}+{\bf u})\right).
$$
It follows that $\sigma_{{\bf u},A} \in \Aut(C)$, where $C = \C = \Ev \left({\mathcal F}(\ell,m)\right)$.
\end{proof} 

Observe that $\phi_{\mathbf{0}, I}$ is the identity transformation of $\Aff^{\delta}$, where $\mathbf{0}$ denotes the zero matrix in 
$M_{\ell \times \ell'}(\F)$ and $I$ the identity matrix in $\GL_{\ell'}(\F)$. Moreover, given any $A, B \in \GL_{\ell'}(\F)$ and 
${\bf u}, \mathbf{v} \in M_{\ell \times \ell'}(\F)$, we have  
\begin{equation}\label{eq:phiAB}
\phi_{{\mathbf{u}},A}\circ\phi_{{\mathbf{v}},B}= \phi_{\mathbf{w},AB} 
\quad \mbox{and} \quad \phi_{{\mathbf{u}},A}^{-1} = \phi_{\mathbf{u'},A^{-1}},
\end{equation}
where $\mathbf{w}:= {\mathbf{v}}A^{-1}+{\mathbf{u}}$ and $\mathbf{u'} = -\mathbf{u}A$. This leads to the following observation-cum-definition. 

\begin{definition} 
\label{def:Slm}
The set 
$
\{\phi_{{\mathbf{u}},A} \, : \, A \in \GL_{\ell'}(\F) \mbox{ and }  {\mathbf{u}}\in M_{\ell\times\ell'}(\F)\}
$
forms a group with respect to composition of maps and this group will be denoted by ${\mathfrak G}(\ell,m)$.
\end{definition}

We determine the group structure of ${\mathfrak G}(\ell,m)$ in the following proposition.

\begin{proposition}\label{prop:groupstructure}
As a group ${\mathfrak G}(\ell,m)$ is isomorphic to the semidirect product $M_{\ell\times\ell'}(\F) \rtimes_{\theta} \GL_{\ell'}(\F)$, where the homomorphism $\theta: GL_{\ell'}(\F) \to Aut(M_{\ell\times\ell'}(\F))$ is defined by $\theta(A)(B):=BA^{-1}$.
\end{proposition}
\begin{proof}
Recall that if $G$ and $H$ are any groups, and if $\theta: H\to \Aut(G)$ is any group homomorphism, then the semidirect product $G\rtimes_{\theta} H$ of $G$ and $H$ relative to $\theta$ is the group whose underlying set is $G \times H$ 
and whose group operation is defined by $(g,h)(g',h')=\left(g\theta(h)(g'), \, hh'\right)$. In our case, $G$ is the additive group $M_{\ell\times\ell'}(\F)$ and $H$ is the multiplicative group $\GL_{\ell'}(\F)$, while
$\theta: H\to \Aut(G)$ is given by $\theta(A)({\mathbf{u}}):={\mathbf{u}}A^{-1}$. 
Now observe that $\theta(A) \in \Aut (G)$ for all $A\in H$ and $\theta(A_1A_2)=\theta(A_1)\theta(A_2)$ for all $A_1, A_2\in H$.  So $\theta$ is indeed
a homomorphism of $H$ into $\Aut (G)$. 
Moreover, in view of \eqref{eq:phiAB}, the group operation $({\mathbf{u}},A)({\mathbf{v}},B)=\left({\mathbf{u}}+{\mathbf{v}}A^{-1},AB\right)$ in $G\rtimes_{\theta} H$ is consistent with the group operation in ${\mathfrak G}(\ell,m)$. 
Thus $({\mathbf{u}},A)\mapsto \phi_{{\mathbf{u}},A}$ gives an isomorphism of 
$M_{\ell\times\ell'}(\F) \rtimes_{\theta} GL_{\ell'}(\F)$ onto ${\mathfrak G}(\ell,m)$. 
\end{proof}


\begin{theorem}
\label{thm:autom}
The automorphism group of the affine Grassmann code $\C$ contains a subgroup isomorphic to ${\mathfrak G}(\ell,m)$. In particular, 
$\# \Aut\left(\C\right)$ is greater than or equal to the quantity in \eqref{eq:orderAut}. 
\end{theorem}
\begin{proof}
In view of Lemma \ref{lem:sigmauA}, $\phi_{{\mathbf{u}},A} \mapsto \sigma_{{\mathbf{u}},A}$ gives a natural map from ${\mathfrak G}(\ell,m)$ into $\Aut\left(\C\right)$. It is readily seen that this map is a group homomorphism. 
So it suffices to show that this homomorphism is injective. To this end, suppose $\sigma_{{\mathbf{u}},A}$ is the identity permutation for some 
$A \in \GL_{\ell'}(\F)$ and  ${\mathbf{u}}\in M_{\ell\times\ell'}(\F)$.  
Then $\sigma_{{\mathbf{u}},A}(\Ev(f))=\Ev(f)$ for all $f \in {\mathcal F}(\ell,m)$, i.e.,
$$
f(PA^{-1}+{\mathbf{u}})=f(P) \quad \mbox{ for all $f \in {\mathcal F}(\ell,m)$ and all $P \in \Adelta$}.
$$
By choosing $P$ to be the zero matrix and letting $f$ vary over all possible $1\times 1$ minors, we find that ${\mathbf{u}}=0$. Further, by choosing $P=e_{ij}$, i.e., $P$ to be the $\ell\times\ell'$ matrix with $1$ in $(i,j)$-th position and $0$ elsewhere,
and again letting $f$ vary over all possible $1\times 1$ minors, we see that $A^{-1}$ is the identity matrix $I$. Hence $A=I$. 
\end{proof}

We leave the question of the complete determination of the automorphism group $\Aut(\C)$ open for future investigation. 

\section{Characterization of minimum weight codewords}
\label{sec:char}

In Section \ref{sec:mindist}, we have calculated the minimum distance $d(\ell,m)$ of the affine Grassmann code $\C$. In this section, we will 
give an explicit characterization of all of its codewords of weight $d(\ell,m)$. 
One of the tools utilized will be a concept involving the specialization function $s_{\mathbf{a}}^{(i)}$ from Definition \ref{def:spec}, which is defined below. 

\begin{definition}
Let $f \in {\mathcal F}(\ell,m)$ and let $i$ be an integer between $1$ and $\ell$. The $i$-th row-vanishing locus of $f$ is the set  $$
V_f^{(i)} :=\{{\mathbf{a}} \in \F^{\ell'} \, : \, s_{\mathbf{a}}^{(i)}(f)=0\}.
$$
\end{definition}

It turns out that if a polynomial in ${\mathcal F}(\ell,m)$ is changed by a translation of the underlying matrix $X$ to $X+\mathbf{u}$, 
then its $i$-th row-vanishing locus is a translate of the corresponding locus of the transformed polynomial 
by the $i$-th row of $\mathbf{u}$.    

\begin{lemma}\label{lemma:visplane}
Let $f \in {\mathcal F}(\ell,m)$ and let $i$ be an integer between $1$ and $\ell$. Then 
$$
V_f^{(i)}= {\mathbf{u}}_i + V_{\phi_{{\mathbf{u}},I}(f)}^{(i)} \quad \mbox{ for every } {\mathbf{u}} \in M_{\ell\times\ell'}(\F),
$$
where $I$ denotes the identity matrix in $\GL_{\ell'}(\F)$. 
\end{lemma}
\begin{proof}
Let ${\mathbf{u}} \in M_{\ell\times\ell'}(\F)$ and let $g:= \phi_{{\mathbf{u}},I}(f)$. 
Suppose $\mathbf{a}\in V_f^{(i)}$. Define $\mathbf{b}\in \F^{\ell'}$ by the relation $\mathbf{a}= \mathbf{u}_i + \mathbf{b}$.  Note that 
\begin{equation}
\label{eq:sbi}
s_{\mathbf{b}}^{(i)}\left(g\right) = g(X)|_{\mathbf{X}_i = \mathbf{b}} = f(X+ \mathbf{u})|_{\mathbf{X}_i = \mathbf{b}}.
\end{equation}  
Now $s_{\mathbf{a}}^{(i)}(f) = f|_{\mathbf{X}_i = \mathbf{a}} =0$. In particular, the polynomial $f|_{\mathbf{X}_i = \mathbf{a}}$ evaluates 
to $0$ for every specialization of the rows $\mathbf{X}_1, \dots, \mathbf{X}_{i-1}, \mathbf{X}_{i+1}, \dots , \mathbf{X}_{\ell}$ to arbitrary vectors in $\F^{\ell'}$.  Since translations by a fixed vector in $\F^{\ell'}$ give a bijection of $\F^{\ell'}$ into itself, this implies that $g|_{\mathbf{X}_i = \mathbf{b}}$  evaluates to $0$ for every specialization of the rows $\mathbf{X}_1, \dots, \mathbf{X}_{i-1}, \mathbf{X}_{i+1}, \dots , \mathbf{X}_{\ell}$ to arbitrary vectors in $\F^{\ell'}$. Hence by the injectivity of the evaluation map $\Ev: {\mathcal F}(\ell-1,m-1)\to \F^{(\ell-1)\ell'}$, we see that $g|_{\mathbf{X}_i = \mathbf{b}}$ is the zero polynomial. Thus, in view of \eqref{eq:sbi}, $\mathbf{b}\in V_g^{(i)}$, 
i.e., $\mathbf{a}\in {\mathbf{u}}_i + V_f^{(i)}$. This proves that $V_f^{(i)} \subseteq {\mathbf{u}}_i + V_g^{(i)}$. The reverse inclusion is proved similarly. 
\end{proof}

\begin{corollary}\label{cor:rowvanishingaffine}
Let $f \in {\mathcal F}(\ell,m)$ and let $i$ be an integer between $1$ and $\ell$. Then the $i$-th row-vanishing locus is either empty or an affine linear space over $\F$, i.e., either $V_f^{(i)}  =\emptyset$ or $V_f^{(i)}  = \mathbf{a}+V$ for some $\mathbf{a}\in \F^{\ell'}$ and a $\F$-linear space $V$. 
\end{corollary}
\begin{proof}
Suppose $V_f^{(i)} \ne \emptyset$. Then there exists some $\mathbf{a}\in V_f^{(i)}$. Let $\mathbf{u} \in M_{\ell\times\ell'}(\F)$ be such that $\mathbf{u}_i =\mathbf{a}$ and $\mathbf{u}_j = \mathbf{0}$ for $1\le j \le \ell$ with $j\ne i$. Also let $g:= \phi_{{\mathbf{u}},I}(f)$. Then 
by Lemma  \ref{lemma:visplane}, $V_f^{(i)}  = \mathbf{a}+V_g^{(i)}$. It remains to show that  $V_g^{(i)}$ is a subspace of $\F^{\ell'}$. To this end, first note that $\mathbf{0} \in V_g^{(i)}$, thanks to the choice of $\mathbf{a}$. Now observe that for any minor ${\mathcal M}\in \Delta (\ell, m)$, we have 
$s_{\bf 0}^{(i)}(\mathcal M)=0$ if $\mathcal M$ involves the $i$-th row and $s_{\bf 0}^{(i)}(\mathcal M)=\mathcal M$ otherwise. Since
$s_{\bf 0}^{(i)}(g)=0$, Lemma \ref{lem:dimFml} implies that $g$ is a $\F$-linear combination of minors of $X$ that involve the $i$-th row. Hence using the multilinearity of the determinant, we readily see that $V_g^{(i)}$ is closed under addition and scalar multiplication. 
\end{proof}

The following result is an analogue of Lemma \ref{lemma:visplane} for homogeneous linear transformations of the underlying matrix. 

\begin{lemma}\label{lemma:homogvisplane}
Let $f \in {\mathcal F}(\ell,m)$ and let $i$ be an integer between $1$ and $\ell$. Then 
$$
V_{\phi_{{\mathbf{0}},A}(f)}^{(i)} = V_f^{(i)}A := \left\{{\mathbf{a}}A :\mathbf{a}\in V_f^{(i)} \right\} \quad \mbox{ for every } A \in \GL_{\ell'}(\F),
$$
where $\mathbf{0}$ denotes the zero matrix in $M_{\ell\times\ell'}(\F)$. 
\end{lemma}
\begin{proof}
Let $A \in \GL_{\ell'}(\F)$. Consider $h:= \phi_{{\mathbf{0}},A}(f)$., i.e., $h\in {\mathcal F}(\ell,m)$ given by $h(X) = f(XA^{-1})$.  
Observe that if, as before, $\mathbf{X}_1, \dots , \mathbf{X}_{\ell}$ denote the row vectors of $X$, then $\mathbf{X}_1A^{-1}, \dots , \mathbf{X}_{\ell}A^{-1}$ are the row-vectors of $XA^{-1}$. Thus the specialization $\mathbf{X}_i = \mathbf{a}A$ in $h$  corresponds to the specialization $\mathbf{X}_i = \mathbf{a}$ in $f$. The rest of the proof is similar to that of Lemma \ref{lemma:visplane}.
\end{proof}
 
Using the row-vanishing locus, one can obtain a useful estimate for the Hamming weight of a codeword from $\C$.

\begin{proposition}\label{prop:rowspec}
Let $f \in {\mathcal F}(\ell,m)$ and let $i$ be an integer between $1$ and $\ell$. Suppose 
$t=\#V_f^{(i)}$. Then 
\begin{equation}
\label{eq:testimate}
\w(\Ev(f)) \ge \frac{q^{\ell'}-t}{q^{\ell'}-q^{\ell'-\ell}}d(\ell,m).
\end{equation}
\end{proposition}
\begin{proof}
In view of Lemma \ref{lem:spec} and the definition of $V_f^{(i)}$, we see that 
\begin{equation}\label{eq:Specwitht}
\w(\Ev(f))=\sum_{{\mathbf{a}} \in \F^{\ell'}\setminus V_f^{(i)}}\w(\Ev(s_{\mathbf{a}}^{(i)}(f))) \ge  \big(q^{\ell'}-t \big)d(\ell-1, m-1).
\end{equation}
On the other hand, by Theorem \ref{thh:mindist},
$$
d(\ell,m) = q^{\ell(\ell'-\ell)}\prod_{i=0}^{\ell-1}(q^\ell-q^i) \quad \mbox{and} \quad 
d(\ell-1,m-1) = q^{(\ell-1)(\ell'-\ell)}\prod_{j=0}^{\ell-2}(q^{\ell-1}-q^j).
$$
Hence, by a direct computation, 
$d(\ell,m)/d(\ell-1, m-1)=q^{\ell'}-q^{\ell'-\ell}$. Combining this with \eqref{eq:Specwitht}, we obtain 
the desired result. 
\end{proof}

Proposition \ref{prop:rowspec} has the following important corollary for minimum weight codewords, which will be the key to our characterization 
of minimum weight codewords.

\begin{corollary}\label{cor:vanishing}
Let $f \in {\mathcal F}(\ell,m)$. If $\Ev(f)$ is a minimum weight codeword of $\C$, then $\#V_f^{(i)} \ge q^{\ell'-\ell}$ 
for all $i \in \{1, \dots , \ell\}$. 
\end{corollary}
\begin{proof}
If $\#V_f^{(i)} < q^{\ell'-\ell}$ for some $i\in \{1, \dots , \ell\}$, 
then by 
Proposition \ref{prop:rowspec}, we obtain 
$\w(\Ev(f))>d(\ell,m)$.
\end{proof}

We are now ready to formulate and prove a characterization of minimum weight codewords of $\C$. 
Recall that if $Y=\left(Y_{ij}\right)$ is any $\ell\times \ell'$ matrix and, as before, $\ell \le \ell'$, then the 
\emph{leading maximal minor} of $Y$ is the minor formed by the first $\ell$ columns of $Y$, namely, 
$\det\big(\left(Y_{ij}\right)_{1\le i,j \le \ell}\big)$.

\begin{theorem}\label{thm:char}
Let $f \in {\mathcal F}(\ell,m)$. Then $f$ is a minimum weight codeword of $\C$  if and only if $f$ is in the ${\mathfrak G}(\ell,m)$-orbit of the leading maximal minor of $X$. 
In other words, $\w(\Ev(f))=d(\ell,m)$ if and only if $f$ is the leading maximal minor of $Y$, where $Y=XA^{-1}+\mathbf{u}$ for some $A \in \GL_{\ell'}(\F)$ and $\mathbf{u} \in M_{\ell \times \ell'}(\F)$. 
\end{theorem}
\begin{proof}
Let ${\mathcal L} := \det\big(\left(X_{ij}\right)_{1\le i,j \le \ell}\big)$ denote the leading maximal minor of $X$. Suppose $f$ is in the 
${\mathfrak G}(\ell,m)$-orbit of ${\mathcal L}$. Then, as noted in Section \ref{sec:automorph}, the codewords $\Ev(f)$ and $\Ev({\mathcal L})$ differ from each other by a permutation of the coordinates. Hence 
$\w(\Ev(f))=\w(\Ev({\mathcal L})) =d(\ell,m)$, thanks to  Lemma \ref{lem:upper}. 

To prove the converse, suppose $\w(\Ev(f))=d(\ell,m)$. Then $f$ must be a nonzero polynomial since $d(\ell,m)<q^{\delta}$. Further, since $\ell'-\ell\ge 0$, Corollary \ref{cor:vanishing} implies that $V_f^{(i)}$ is nonempty for each $i\in \{1, \dots , \ell\}$. Choose ${\mathbf{u}}_i \in V_f^{(i)}$ for $1\le i \le \ell$. Let ${\mathbf{u}} \in M_{\ell\times\ell'}(\F)$ be 
the $\ell\times\ell'$ matrix whose $i$-th row vector is ${\mathbf{u}}_i$ for $1\le i \le \ell$, and let $g:= \phi_{{\mathbf{u}},I}(f)$. 
Then $g$ is in the ${\mathfrak G}(\ell,m)$-orbit of $f$ and by Lemma \ref{lemma:visplane}, 
$$
V_f^{(i)}={\mathbf{u}}_i+V_{g}^{(i)} \quad \mbox{and}\quad \mathbf{0}\in V_{g}^{(i)} \quad \mbox{for each } i\in \{1, \dots , \ell\}.
$$
Thus, $s_{\bf 0}^{(i)}(g)=0$ for each $i\in \{1, \dots , \ell\}$. Now observe that for any ${\mathcal M} \in \Delta(\ell,m)$  
and any $i\in \{1, \dots , \ell\}$, we have  $s_{\bf 0}^{(i)}({\mathcal M})=0$ if ${\mathcal M}$ involves the $i$-th row of $X$ and $s_{\bf 0}^{(i)}({\mathcal M})={\mathcal M}$ otherwise. Consequently, if $g=\sum_{{\mathcal M} \in \Delta(\ell,m)}a_{{\mathcal M}}{\mathcal M}$, where
$a_{{\mathcal M}}\in \F$ for ${\mathcal M}\in \Delta(\ell,m)$, then by Lemma \ref{lem:dimFml}, we see that $a_{{\mathcal M}} =0$ 
for all ${\mathcal M} \in \cup_{i=0}^{\ell-1}\Delta_i(\ell,m)$. 
This proves that $g$ is a $\F$-linear combination of $\ell \times \ell$ minors of $X$. In particular, if $\ell'=\ell$, then $\mathcal L$ being 
the only $\ell \times \ell$ minors of $X$, we obtain $g= c{\mathcal L}$ for some
$c\in \F$ with $c\ne 0$. Since ${\mathcal L} = \phi_{\mathbf{0}, D}(c{\mathcal L})$, where $D$ denotes the $\ell'\times \ell'$ diagonal matrix ${\rm diag}\left(c, 1, \dots , 1\right)$ in $\GL_{\ell'}(\F)$, we see that $f$ is in the ${\mathfrak G}(\ell,m)$-orbit of ${\mathcal L}$ when $\ell'=\ell$. 

Now suppose $\ell< \ell'$.  Consider the first row-vanishing space  $V_g^{(1)}$. In view of Corollary \ref{cor:rowvanishingaffine} 
and the fact that $\mathbf{0}\in V_{g}^{(1)}$, we see that $V_g^{(1)}$ is a linear space over $\F$. Moreover, Corollary \ref{cor:vanishing} 
implies that the dimension of $V_g^{(1)}$ is at least $\ell'-\ell$. Hence we can choose linearly independent vectors ${\mathbf{b}}_1,\dots,{\mathbf{b}}_{\ell'-\ell} \in V_g^{(1)}$. Let ${\mathbf{b}}$ be
the $(\ell'-\ell)\times\ell'$ matrix whose $i$-th row vector is ${\mathbf{b}}_i$ for ${1\le i \le \ell'-\ell}$. Since ${\mathbf{b}}$ has full rank, there exists an invertible matrix $A \in \GL_{\ell'}(\F)$ such that
$$
{\mathbf{b}}A= \left(\mathbf{0}_{(\ell'-\ell)\times \ell} \, | \, I_{\ell'-\ell}\right) = 
\left(
\begin{array}{ccccccc}
0      & \cdots & 0      & 1      & 0 & \dots  & 0 \\
0      & \cdots & 0      & 0      & 1 & \dots & 0 \\ 
\vdots &        & \vdots & \vdots &   & \ddots & \vdots\\
0      & \cdots & 0      & 0      & 0 & \dots  & 1 \\
\end{array}
\right).
$$
Indeed, the matrix on the right is essentially the reduced column-echelon form of ${\mathbf{b}}$. 
We now consider the function $h=h(X)=\phi_{{\bf 0},A}(g(X))=g(XA^{-1})$. Clearly, $h$ is in the ${\mathfrak G}(\ell,m)$-orbit of $g$ and hence of $f$; in particular, $\w(\Ev(h)) = d(\ell,m)$ and $h$ is a nonzero polynomial. 
By the multilinearity of the determinant, it can, just as $g$, be written as a $\F$-linear combination of $\ell \times \ell$ minors of $X$. 
For $1\le j \le \ell'$, let $\mathbf{e}_j$ denote the vector in $\F^{\ell'}$ with $1$ in the $j$-th position and $0$ elsewhere. Observe that 
if ${\mathcal M} \in \Delta_\ell(\ell,m)$ is the minor formed by the columns of $X$ indexed by $j_1, \dots , j_{\ell}$, where $1\le j_1<\cdots<j_{\ell}\le \ell'$, then 
$s^{(1)}_{\mathbf{e}_j} ({\mathcal M})= 0$ if $j\not\in\{j_1,\dots,j_\ell\}$, whereas $s^{(1)}_{\mathbf{e}_j}({\mathcal M})$ is a nonzero polynomial (and, in fact, $\pm {\mathcal M}_1$, where ${\mathcal M}_1$ is a $(\ell-1)\times(\ell-1)$ minor of $X$) if $j\in\{j_1,\dots,j_\ell\}$. 
Now by the choice of $A$ and by Lemma 25, we have that $\mathbf{e}_j \in V_h^{(1)}$ for all $j$ such that $\ell < j \le \ell'$. 
Consequently, if $h=\sum_{{\mathcal M} \in \Delta_{\ell}(\ell,m)} a_{{\mathcal M}}{\mathcal M}$, where $a_{{\mathcal M}}\in \F$ for ${\mathcal M}\in \Delta_{\ell}(\ell,m)$, then by Lemma \ref{lem:dimFml}, we see that 
$a_{{\mathcal M}} =0$ for all those ${\mathcal M}$ in $\Delta_{\ell}(\ell,m)$ that involve the $j$-th column of $X$ for some $j>\ell$. But the only $\ell\times \ell$ minor of $X$ that does not involve the $j$-th column of $X$ for some $j>\ell$ is $\mathcal L$. Hence  $h= c{\mathcal L}$ for some $c\in \F$ with $c\ne 0$. It follows that $f$ is in the ${\mathfrak G}(\ell,m)$-orbit of ${\mathcal L}$. 
\end{proof}

In case $\ell'=\ell$, the above theorem simplifies to the statement in Remark \ref{rem:mis2l}.

\section{Enumeration of minimum weight codewords}
\label{sec:enum}

In this section, we let $d=d(\ell,m)$ denote the minimum distance of $\C$ and $A_d$ the number of minimum weight codewords of $\C$. Having
characterized the codewords of weight $d$ in the previous section, we now proceed to compute $A_d$. Equivalently, we determine the number of polynomials $f \in {\mathcal F}(\ell,m)$ giving rise to minimum weight codewords. We have seen in Section \ref{sec:automorph} that the
finite group ${\mathfrak G}(\ell,m)$ acts naturally on ${\mathcal F}(\ell,m)$.  With this in view, we can use standard group theory together 
with Theorem \ref{thm:char} to obtain the following.

\begin{lemma}\label{lem:groupaction}
Let ${\mathcal L}=\det\big((X_{ij})_{1\le i,j \le \ell}\big)$ be the leading maximal minor of $X$. Then 
$$
A_{d}=\frac{\#{\mathfrak G}(\ell,m)}{\#\stab({\mathcal L})},
$$ 
where $\stab({\mathcal L})$ denotes the stabilizer of the minor ${\mathcal L}$.
\end{lemma}
\begin{proof}
By Theorem \ref{thm:char}, the cardinality of the ${\mathfrak G}(\ell,m)$-orbit of ${\mathcal L}$ is equal to $A_d$. On the other hand, for any finite group acting on a finite set, the cardinality of the orbit of an element is equal to the index of its stabilizer. 
\end{proof}

Thanks to Lemma \ref{lem:groupaction}, the computation of $A_d$ reduces to the problem of finding 
the cardinality of the stabilizer of ${\mathcal L}:= \det\big((X_{ij})_{1\le i,j \le \ell}\big)$. 
To this end, let us begin by observing that if $f \in {\mathcal F}(\ell,m)$ is in the ${\mathfrak G}(\ell,m)$-orbit of ${\mathcal L}$, 
i.e., if $f= \phi_{\mathbf{u},A}({\mathcal L})$ for some $A\in \GL_{\ell'}(\F)$ and $\mathbf{u}\in M_{\ell\times \ell'}(\F)$,  then 
\begin{equation}
\label{eq:orbitofL}
f = \det (XM+\mathbf{m}) \quad \mbox{for some $M \in M_{\ell'\times\ell}(\F)$ of rank $\ell$ and ${\bf m} \in M_{\ell\times\ell}(\F)$.} 
\end{equation}
Indeed, it suffices to take $M$ to be the ${\ell'\times\ell}$ matrix formed by the first $\ell$ columns of $A^{-1}$ and ${\bf m}$ to be the ${\ell\times\ell}$ matrix formed by the first $\ell$ columns of $\mathbf{u}$, and observe that $\rank(M)= \ell$ since $A$ is nonsingular and that the leading maximal minor of the $\ell\times \ell'$ matrix $XA^{-1}+\mathbf{u}$ is $\det (XM+\mathbf{m})$. We shall now analyze when a polynomial $f$ given by \eqref{eq:orbitofL} is in the stabilizer of $\mathcal L$. As usual, we denote by $\SL_\ell(\F)$ the special linear 
group of $\ell \times \ell$ matrices over $\F$, viz., $\SL_\ell(\F):=\{A\in \GL_\ell(\F) : \det A = 1\}$. 
 
\begin{lemma}\label{lem:stab}
Let ${\mathcal L}=\det\big((X_{ij})_{1\le i,j \le \ell}\big)$ be the leading maximal minor of $X$. Also  
let $M \in M_{\ell'\times\ell}(\F)$ be of rank $\ell$ and ${\bf m} \in M_{\ell\times\ell}(\F)$. 
Then ${\mathcal L}=\det(XM+{\bf m})$ if and only if ${\bf m}={\bf 0}$ and 
there exists $E \in \SL_\ell(\F)$ such that the first $\ell$ rows of $ME$ form the $\ell\times \ell$ identity matrix, while the last $\ell'-\ell$ 
rows are zero. In this case, the matrix $E$ in $\SL_\ell(\F)$ is uniquely determined by $M$.
\end{lemma}
\begin{proof}
We start by showing the uniqueness of the matrix $E$. Suppose 
$$
ME_1=\left(\begin{array}{l}{\mathbf{I}_{\ell}}\\ \mathbf{0} \end{array}\right) =ME_2 \quad \mbox{ for some $E_1, E_2\in \SL_\ell(\F)$,}
$$
where ${\mathbf{I}_{\ell}}$ denotes  the $\ell\times \ell$ identity matrix and $\mathbf{0}$ the $(\ell'-\ell)\times \ell$ zero matrix.
Then $M(E_2-E_1)=0$. Since $M$ has full rank, this can only happen if $E_1=E_2$.

To prove the equivalence, first suppose there exists $E \in \SL_\ell(\F)$ such that  
\begin{equation}\label{eq:mandME}
ME=\left(\begin{array}{l}{\mathbf{I}_{\ell}}\\ \mathbf{0} \end{array}\right), 
\end{equation}
and also suppose $\mathbf{m} = \mathbf{0}$. 
Then 
$$
\det(XM+{\bf m})= \det(XM)=\det(XME)= \det \left( X \left(\begin{array}{l}{\mathbf{I}_{\ell}}\\ \mathbf{0} \end{array}\right) \right) = {\mathcal L}.
$$

Conversely, suppose ${\mathcal L}=\det(XM+{\bf m})$. Since $M$ has full rank, there exists 
$N\in M_{\ell'\times\ell}(\F)$ such that $NM={\bf m}$. 
Hence ${\mathcal L} =\det(XM+{\bf m})=\det((X+N)M)$. Using Cauchy-Binet formula (Lemma \ref{lem:CB}) and the notation therein, we now find 
\begin{equation}\label{eq:CBdecomp}
{\mathcal L} =\sum_{I}\det((X+N)^I)\det(M_I),
\end{equation}
where the sum is over all subsets $I$ of $\{1,\dots,\ell'\}$ of cardinality $\ell$. For any such~$I$, Lemma \ref{lem:xplusa} implies that 
$\det((X+N)^I)$ is the sum of $\det(X^I)$ and a $\F$-linear combination of minors of $X^I$ of order $<\ell$. Hence, comparing terms of total degree $\ell$ in \eqref{eq:CBdecomp}, 
we obtain 
\begin{equation}\label{eq:CBhom}
{\mathcal L}= \det(X^{I^*}) = \sum_{I}\det(X^I)\det(M_I), \quad \mbox{where} \quad I^*:=\{1, \dots , \ell\}.
\end{equation}
Consequently, in view of Lemma \ref{lem:dimFml}, $\det(M_{I^*})=1$, while $\det(M_I)=0$ for every $I\subseteq \{1,\dots,\ell'\}$ with $\# I = \ell$ and $I\ne I^*$. Define $E:= M_{I^*}^{-1}$. It is clear that $E \in \SL_\ell(\F)$. Moreover, by the choice of $E$, the first $\ell$ rows of $ME$ form the $\ell\times\ell$ identity matrix ${\mathbf{I}_{\ell}}$. We claim that for any $i>\ell$, the $i$-th row $ME_i$ of $ME$ is zero. To see this, write $ME_i=(b_1,\dots,b_\ell)$. Choose any $j \in I^*$ 
and let $I:=\left(I^*\cup\{i\} \right)\setminus\{j\}$. 
Then $I\subseteq \{1,\dots,\ell'\}$ with $\# I = \ell$ and 
$\det(M_I)=0$ since $I\ne I^*$. On the other hand, $\det(M_I)=\det(M_IE)$. Now, since the first $\ell-1$ elements of $I$ are contained in $\{1,\dots,\ell\}$, the first $\ell-1$ rows of the matrix $M_IE$ form the matrix obtained from ${\mathbf{I}_{\ell}}$ 
by deleting its $j$-th row. This implies that $0=\det(M_IE)=\pm (M_IE)_{\ell j} = \pm (ME)_{ij} $. By varying $j$ over $I^*$, 
we obtain $ME_i=(0,\dots,0)$. This proves the claim. 
It remains to show that ${\bf m}={\bf 0}$. 
We have noted earlier that there is $N\in M_{\ell \times \ell'}(\F)$ such that $\mathbf{m} =NM$. Hence 
$$ 
{\mathcal L}=\det(XM+{\bf m})=\det((X+N)M)=\det((X+N)ME)=\det((X_{ij}+N_{ij})_{1\le i,j \le \ell}),
$$
where the penultimate equality follows since $E\in \SL_\ell(\F)$ and the last equality follows since $ME$ satisfies \eqref{eq:mandME}. 
Using Lemma \ref{lem:xplusa} together with Lemma \ref{lem:dimFml}, by comparing the coefficients of $(\ell-1)\times(\ell-1)$ minors, we find  $N_{ij}=0$ for $1\le i,j \le \ell$. But then ${\bf m}E=N(ME)=(N_{ij})_{1\le i,j \le \ell}={\bf 0}$, thanks to \eqref{eq:mandME}. 
Since $E$ is invertible, this implies that ${\bf m}={\bf 0}$.
\end{proof}

We are now ready to compute the cardinality of the stabilizer of the leading maximal minor.

\begin{lemma}\label{prop:stabcount}
Let ${\mathcal L}=\det\big((X_{ij})_{1\le i,j \le \ell}\big)$ be the leading maximal minor of $X$. Then 
$$
\#\stab({\mathcal L})=\frac{q^{\ell(\ell'-\ell)}}{q-1}\prod_{i=\ell}^{\ell'-1}(q^{\ell'}-q^{i})\prod_{j=0}^{\ell-1}(q^\ell-q^i). 
$$
\end{lemma}
\begin{proof} 
Let $A\in\GL_{\ell'}(\F)$ and ${\mathbf{u}}\in M_{\ell\times\ell'}(\F)$. Suppose $\phi_{{\mathbf{u}},A}({\mathcal L})={\mathcal L}$. First we write $A=(M \, | \, R)$ where %
$M\in M_{\ell'\times\ell}(\F)$ and 
$R\in M_{\ell'\times(\ell'-\ell)}(\F)$ are matrices formed, respectively, by the first $\ell$ columns of $A$  
and the remaining $\ell'-\ell$ columns of $A$.  
Similarly, we write ${\mathbf{u}}=({\bf m} \, |\, {\bf r})$. Then, as in \eqref{eq:orbitofL}, 
${\mathcal L} = \det(XM+{\bf m})$. Hence by Lemma \ref{lem:stab}, 
${\bf m}={\bf 0}$ and moreover, there exists a unique 
$E \in SL_\ell(\F)$ such that 
\begin{equation}\label{eq:stab}
A\left(\begin{array}{c|c} E & {\mathbf{0}} \\ \hline {\mathbf{0}} & I_{\ell'-\ell}\\  \end{array}\right)= A\left(M \, | \, R \right) = 
\left(\begin{array}{c|} {\mathbf{I}_{\ell}} \\ \hline {\mathbf{0}}\\  \end{array}\begin{array}{c}R\end{array}\right),
\end{equation}
where $ {\mathbf{0}}$ denotes the zero matrix of an appropriate size and, as before, ${\mathbf{I}_{\ell}}$ denotes the $\ell\times \ell$ 
identity matrix. The matrices $R$ and ${\bf r}$ do not have any effect on $\phi_{{\mathbf{u}},A}({\mathcal L})$ and can therefore be chosen freely. However, $R$ has to be chosen in such a way that the matrix on the right hand side of \eqref{eq:stab} has full rank. This means that the last 
$\ell'-\ell$ rows of $R$ must be linearly independent. It follows that $\#\stab({\mathcal L})$ is the product of $\#\SL_\ell(\F)$ 
and the following terms:
$$
\begin{tabular}{ll}
$q^{\ell(\ell'-\ell)},$   & for the choice of $\bf r$,\\
$q^{\ell(\ell'-\ell)},$   & for the choice of the first $\ell$ rows of $R$, and \\
${\prod_{i=0}^{\ell'-\ell-1}(q^{\ell'-\ell}-q^i)},$ & for the choice of the last $\ell'-\ell$ rows of $R$.
\end{tabular}
$$
Since $\#\SL_\ell(\F)=(q-1)^{-1}\prod_{j=0}^{\ell-1}(q^{\ell}-q^j)$, the lemma is proved.
\end{proof}

We now obtain the main result of this section concerning the number of codewords of $\C$ of weight $d(\ell,m)$. 
The result is best formulated using the Gaussian binomial coefficient defined, for any integers $k$ and $n$ with $1\le k \le n$, as follows.  
\begin{equation}\label{eq:GB}
{{n}\brack{k}}_q := \frac{[n]_q!}{[k]_q![n-k]_q!} = \frac{(q^n-1)(q^n-q)\dots (q^n-q^{k-1})}{(q^k-1)(q^k-q)\dots (q^k-q^{k-1})}.
\end{equation}
It is well-known that \eqref{eq:GB} is a monic polynomial in $q$ of degree $k(n-k)$ with nonnegative integral coefficients. 
In particular, ${{n}\brack{n}}_q = 1$.  
\begin{theorem}
\label{thm:enum}
The number $A_d$ of codewords of weight $d(\ell,m)$ of the affine Grassmann code $\C$  is given by  
$$
A_d=(q-1)q^{\ell^2}{{\ell'}\brack{\ell}}_q . 
$$
\end{theorem}
\begin{proof}
Using Proposition \ref{prop:groupstructure} we see that 
$$
\#{\mathfrak G}(\ell,m)=\#M_{\ell\times\ell'}(\F) \cdot \#GL_{\ell'}(\F)=q^{\ell\ell'}\prod_{i=0}^{\ell'-1}(q^{\ell'}-q^i).
$$
Hence the desired result follows from Lemmas \ref{lem:groupaction} and \ref{prop:stabcount}.
\end{proof}

For $\ell=1$, we obtain $A_d=q(q^{\ell'}-1)$, whereas for $\ell'=\ell$, we obtain $A_d=(q-1)q^{\ell^2}$. This is in agreement with 
Remarks \ref{rem:dethyper} and 
\ref{rem:mis2l}, respectively.

\section{Connection with Grassmann codes}
\label{sec:grassmann}

Grassmann codes, denoted by $C(\ell,m)$, are $[n,k]_q$-linear codes 
defined for any positive integers $\ell,m$ satisfying $1\le \ell\le m$, where 
$$
n :  = {{m}\brack{\ell}}_q  := \frac{(q^m -1) (q^m - q) \dots (q^m - q^{\ell -1})}{(q^\ell  -1) (q^\ell  - q) \dots (q^\ell  - q^{\ell -1})} 
 \quad {\rm and} \quad k : = {\binom{m}{\ell}}.
$$
The case $\ell=m$ is trivial and in general, there is a natural equivalence between $C(\ell,m)$ and $C(m-\ell,m)$. With this in view, we shall assume $1\le \ell < m$ and that $m-\ell \ge \ell$. Thus, if we set $\ell':=m-\ell$, then we have $1\le \ell \le \ell'$ and 
$\ell + \ell'=m$, exactly as in the basic set-up of Sections \ref{sec:prelim} through \ref{sec:enum}.

A quick way to define $C(\ell,m)$ is to say that these are linear codes associated to the projective system obtained from 
the Pl\"ucker embedding of the Grassmann variety $G_{\ell,m}$  in the projective space ${\mathbb P}^{k-1}$ over $\F$. 
Recall that the \emph{Grassmann variety} (also known as the \emph{Grassmannian}) $G_{\ell,m}$ over $\F$ 
is the space of all $\ell$-dimensional subspaces of the $m$-dimensional vector space $\Fqm$ over $\F$. 
The 
{\em Pl\"ucker embedding} maps $ G_{\ell,m}(\Fq)$ into $\PP^{k-1} = \PP(\wedge^{\ell}\Fqm)$ 
by sending a $\ell$-dimensional subspace $W$ spanned by $w_1, \dots , w_\ell$ to the class of $w_1\wedge \dots \wedge w_{\ell}$. 
To obtain this a little more concretely, one can proceed as follows. 
Let 
$$
I(\ell ,m)=\{ \alpha = (\alpha_1, \dots , \alpha_\ell )\in \Z^{\ell}  : 
1\le \alpha_1 < \dots < \alpha_\ell  \le m \} 
$$
be an indexing set [ordered, say, lexicographically] for the points of 
$\PP^{k-1}(\Fq)$. Given any $\alpha \in I(\ell ,m)$ and any 
 $\ell \times m$ matrix $A=(a_{ij})$, let
$$
p_{\alpha} (A) 
= \mbox{ determinant of the $\alpha$-th submatrix of } A : = \ 
\det \left(
 a_{i \alpha_j} \right)_{1\le i, j \le \ell } . 
$$
Now, for any $W \in G_{\ell, m}(\Fq)$, we can find 
a $\ell \times m$ matrix $A_W$ whose rows give a 
basis of $W$, and then 
$$
p (W) = \left( p_{\alpha} (A_W) \right)_{\alpha  \in I(\ell ,m)} \in \PP^{k-1} 
$$
is called the {\em Pl\"ucker coordinate} of $W$. It is easy to see that
this depends only on $W$ and not on the choice of $A_W$. Moreover, 
The map $W\mapsto p(W)$ of $ G_{\ell,m}(\Fq) \to \PP^{k-1}$ 
is precisely the Pl\"ucker embedding; it is well-known that this is injective and its image equals the zero locus of certain quadratic
polynomials. 
Henceforth, we shall identify $W$ with $p(W)$. The definition of $C(\ell,m)$ as the codes corresponding to the projective system 
in $\PP^{k-1}$ given by $ G_{\ell,m}(\Fq)$  
amounts to the following. 

Let ${\mathcal G}(\ell,m) = (\wedge^{\ell}\Fqm)^*$ denote the space of linear forms on $\wedge^{\ell}\Fqm$ [this can be identified with $(\wedge^{m-\ell}\Fqm)$] and let $\{Q_1, \dots , Q_n\}$ be (arbitrary, but fixed, lifts of) points in $\wedge^{\ell}\Fqm$ corresponding to the elements 
of $G_{\ell,m}(\Fq)$  in $\PP^{k-1}$. Now the evaluation map 
$$
\Ev: {\mathcal G}(\ell,m) \to \Aff^n({\F}) \quad \mbox{defined by} \quad \Ev(g):= \left(g(Q_1), \dots , g(Q_n)\right)
$$
is injective (since the Pl\"ucker embedding is nondegenerate) and its image is precisely the Grassmann code $C(\ell,m)$. 

To relate $C(\ell,m)$ to $\C$, let us first note that the projective space $\PP^{k-1}$ is covered by affine spaces $U_{\alpha} \simeq \Aff^{k-1}$, where $U_{\alpha}:= \{p\in \PP^{k-1} : p_{\alpha}=1\}$ and $\alpha$ varies over $I(\ell,m)$. It is a classical fact that the intersection $B_{\alpha} := G_{\ell,m}\cap U_{\alpha}$ is isomorphic to an affine space of dimension $\delta:= \ell \ell' = \ell(m-\ell)$. This isomorphism is described explicitly by the Basic Cell Lemma 
of \cite{GL}.  In effect, if $W\in B_{\alpha}$, then the  $\ell\times m$ matrix $A_W$ associated to $W$ can be chosen in such a way that the $\alpha$-th submatrix of $A_W$ is the identity matrix. Now if $B_W$ denotes the $\ell\times \ell'$ matrix formed by removing from $A_W$ its $\alpha$-th submatrix,  then the entries of $B_W$ can be viewed as variables. Moreover, the $k$-tuple $p(W)$ formed by the $\ell\times \ell$ minors of $A_W$ corresponds to the $k$-tuple formed by arbitrary sized minors of $B_W$. Thus, evaluating linear forms at points of the affine open cell $B_{\alpha}$ of $G_{\ell,m}$ corresponds to evaluating linear forms in arbitrary sized minors of $B_W$ at the points of the $\delta$-dimensional affine space over $\Fq$. In other words, the evaluation map $\Ev: {\mathcal G}(\ell,m) \to \Aff^n$ reduces to the evaluation map on ${\mathcal F}(\ell,m)$ 
considered in Section \ref{sec:prelim}. 

\begin{remark}
\label{rem:affineGrassmannian}
We hope that the above discussion clarifies the genesis of the terminology \emph{affine Grassmann} 
for the codes $\C$ studied in this paper. Indeed, this terminology arises from the fact that in essence, we consider an 
affine open piece of the Grassmann variety instead of the full Grassmann variety. However, 
this terminology should not be confused with the 
so called \emph{affine Grassmannian}, which is usually an infinite dimensional object obtained from the Laurent power series valued points 
of an algebraic group. Indeed, it appears unlikely that interesting and efficient codes could be built from the infinite dimensional affine Grassmannian, and hence there does not seem to be any harm in calling the codes $\C$ as affine Grassmann codes.
\end{remark}

It may be worthwhile to compare the basic parameters of $C(\ell,m)$ and $\C$. This is done in Figure \ref{table} below. 
While the results for $\C$ are proved in the previous sections, those for $C(\ell,m)$ can be found, for example, in \cite{N} and \cite{GL}. 

\bigskip
\begin{figure}[h]
\begin{tabular}{l|l|l}
 & $C(\ell,m)$ &  $\C$ \\ \hline 
 & & \\
\mbox{\rm Length} & ${{m}\brack{\ell}}_q = q^{\delta} + q^{\delta-1}  + 2q^{\delta-2} + \cdots$ & $q^{\delta}$ \\ 
& & \\
Dimension & ${\binom{m}{\ell}}$ & ${\binom{m}{\ell}}$ \\ 
& & \\
Minimum  & $q^{\delta}$ & $q^{\delta-\ell^2}\prod_{i=0}^{\ell-1}(q^\ell-q^i)$ \\
distance & & \quad $= q^{\delta} - q^{\delta-1} - q^{\delta-2} + \cdots$ \\
 & & \\
Number of & $(q-1){{m}\brack{\ell}}_q$  & $(q-1)q^{\ell^2}{{m-\ell}\brack{\ell}}_q$ \\
min. weight & $\quad = O(q^{\delta+1})$  & $\quad = O(q^{\delta+1})$  \\
codewords & & 
\end{tabular}
\caption{A comparison of Grassmann and affine Grassmann codes}
\label{table}
\end{figure}
\bigskip

It may be noted that the two classes of codes are comparable. While the affine Grassmann codes are shorter than Grassmann codes and 
have a better rate, the Grassmann codes fare better in terms of the minimum distance and also the relative distance. In spite of
the connection between the two codes indicated above, there does not seem to be a straightforward way to deduce the properties of one code
directly from that of the other. However, the growing literature on Grassmann codes can provide pointers for further research on affine Grassmann codes, whereas the analogy of affine Grassmann codes with Reed-Muller codes and results obtained in this paper concerning their automorphisms may provide further impetus for the study of Grassmann codes.

\section*{Acknowledgments}
\label{secAck}
\begin{small}
The second named author would like to thank the Department of Mathematics of the Technical University of Denmark for its warm hospitality during his visits in June 2008 and April 2009 when some of this work was carried out.  
\end{small}

\end{document}